\documentclass{article}
\usepackage{ijcai19}

\usepackage{times}
\usepackage{soul}
\usepackage{url}
\usepackage[hidelinks]{hyperref}
\usepackage[utf8]{inputenc}
\usepackage[small]{caption}
\usepackage{graphicx}
\usepackage{amsmath}
\usepackage{amssymb}
\usepackage{relsize}
\usepackage{booktabs}
\usepackage{algorithm}
\usepackage{algorithmic}
\usepackage{subcaption}
\usepackage{graphicx}
\usepackage{multirow}
\usepackage{enumerate}
\usepackage{enumitem}

\usepackage{amsthm}
\newtheorem{theorem}{Theorem}
\newtheorem{corollary}{Corollary}
\newtheorem{lemma}[theorem]{Lemma}
\usepackage{todonotes}

\newcommand{\modelName}[0]{BEBA}
\newcommand{\ptitle}[1]{\vspace{1mm}\noindent{\bf #1.}}

\usepackage{enumitem}
\setlist[itemize]{leftmargin=*}

\title{Opinion Dynamics with Backfire Effect and Biased Assimilation}

\author{
Xi Chen$^1$\footnote{Contact Author}\and
Panayiotis	Tsaparas$^2$\and
Jefrey Lijffijt$^1$\and
Tijl De Bie$^1$\\
\affiliations
$^1$Dept. of Electronics and Information Systems, IDLab, Ghent University\\
$^2$Department of Computer Science and Engineering, University of Ioannina\\
\emails
$^1$ \{firstname.lastname\}@ugent.be \and $^2$ tsap@cs.uoi.gr
}

\begin{document}

\maketitle

\begin{abstract}
The democratization of AI tools for content generation, combined with unrestricted access to mass media for all (e.g. through microblogging and social media), makes it increasingly hard for people to distinguish fact from fiction.
This raises the question of how individual opinions evolve in such a networked environment without grounding in a known reality.
The dominant approach to studying this problem uses simple models from the social sciences on how individuals change their opinions when exposed to their social neighborhood, and applies them on large social networks.

We propose a novel model that incorporates two known social phenomena:
(i) \emph{Biased Assimilation}: the tendency of individuals to adopt other opinions if they are similar to their own; (ii) \emph{Backfire Effect}: the fact that an opposite opinion may further entrench someone in their stance, making their opinion more extreme instead of moderating it. 
To the best of our knowledge this is the first model that captures the Backfire Effect.
A thorough theoretical and empirical analysis of the proposed model reveals intuitive conditions for polarization and consensus to exist, as well as the properties of the resulting opinions.
\end{abstract}

\section{Introduction}\label{sec:introduction}
Recent years have seen an increasing amount of attention from the computational social sciences in the study of opinion formation and polarization over social networks, with applications ranging from politics to brand perception~\cite{
conover2011political,gionis2013opinion,akoglu2014quantifying}.
Much of this research leverages pre-existing opinion formation models that have been studied for decades~\cite{Jackson-book,RevModPhys}.
These models formalize the fact that people form their opinions through interactions with others.
One of the best-known models is DeGroot's model~\cite{degroot1974reaching}, which considers an individual's opinion as dynamic, assuming that it is updated as the weighted average of the individual's current opinion and those of her social neighbors. The weights represent the strength of the social connections.

DeGroot's model is elegant and intuitive and it guarantees that the opinions converge towards a consensus~\cite{degroot1974reaching,Jackson-book}.
Yet, the opinions cannot polarize, contradicting empirical observations~\cite{baron1996social,gilbert2009blogs}.
Variants of DeGroot's model have been proposed that incorporate \textit{biased assimilation}~\cite{krause2000discrete,dandekar2013biased}, which is also known as \textit{confirmation bias} or \textit{myside bias} and refers to the phenomenon where information that corroborates someone's beliefs affects those beliefs more strongly than information that contradicts it~\cite{lord1979biased}. Incorporating biased assimilation has been shown to potentially lead to polarization~\cite{dandekar2013biased} or opinion clustering~\cite{krause2000discrete}.

An extreme manifestation of confirmation bias is a behavior known in social psychology
as the \emph{Backfire Effect}~\cite{nyhan2010corrections,allahverdyan2014opinion}. It refers to the fact that, when an individual is faced with information that contradicts their opinion,
they will not only tend to discredit it, but they will also become more entrenched and thus extreme in their opinion.
The backfire effect may help explain the emergence of polarization. Yet, it has so far been overlooked by existing opinion formation models. 

Motivated by these observations, we propose the {\modelName} model, a novel opinion formation model
that simultaneously models the Backfire Effect and Biased Assimilation.
\modelName{} depends on a single---intuitive, node-dependent---parameter $\beta_i$,
which we call the \emph{entrenchment} of node $i$.
It captures both the tendency of node $i$ to become more entrenched by opposing opinions and the bias towards assimilating opinions favorable to its own.
Our main contributions are:
\begin{itemize}[noitemsep,topsep=0pt,parsep=0pt,partopsep=0pt]
\item We propose the \modelName{} model of opinion formation, which accounts for both the Backfire Effect and Biased Assimilation (Section~\ref{sec:BEBAmodel}). To the best of our knowledge \modelName{} is the first model that incorporates the Backfire Effect.
\item We theoretically analyze the \modelName{} model in Section~\ref{sec:theoreticalAnalysis}, studying conditions for reaching consensus or polarization.
\item In Section~\ref{sec:experiments} we empirically evaluate, on real and synthetic data, the effect of both network topology and initial opinions on polarization / convergence.
\end{itemize}

\section{Related Work}
Opinion formation has been studied in diverse research fields, from psychology and social sciences to economics and physics~\cite{Jackson-book,RevModPhys}. The former mostly use empirical methods to understand the factors that affect opinion formation, while the latter mostly aim to understand emergent behavior implied by these theories.

Two observations from psychology and social sciences relating to our work are the biased assimilation and backfire effect~\cite{corner2012uncertainty,lord2009biased}, which state that individuals are more inclined to accept opinions closer to their own~\cite{lord1979biased}, and that, when exposed to the opposite opinion, individuals entrench themselves in their own opinion~\cite{nyhan2010corrections,chong2007framing,herr1986consequences}, respectively.

We study the common setting where opinions are formalized as real values, formed through social interactions (see~\cite{Jackson-book} and~\cite{RevModPhys} for surveys).
The most popular models include the Voter model~\cite{voter1,voter2}, DeGroot's model~\cite{degroot1974reaching}, and the Friedkin-Johnsen model~\cite{friedkin1990social}. 
Yet, none of these account for the biased assimilation or backfire effect.

There is work on modeling the fact that users are more influenced by opinions closer to their own. The bounded confidence models~\cite{DeffuantNAW00,DeffuantAWF02,Hegselmann02opiniondynamics} assume that a user is influenced only by opinions that are within $\epsilon$ of its own. The work of Kempe et al.,~\cite{Kempe:2016} assumes that there are different types of opinions and users are influenced by opinions of similar types. Das et al.,~\cite{Das:2014} consider a biased version of the voter model that biases individuals to adopt similar opinions.
The work most closely related to ours is that of Dandekar et al.,~\cite{dandekar2013biased} who propose a variant of DeGroot's model to capture the biased assimilation effect. In their model, the importance that a node attaches to the opinion of a neighbor depends on their agreement. However, it does not model the backfire effect.

\section{Model definition}
\label{sec:BEBAmodel}

In this section, we first describe existing models on which our work builds and then introduce our nonlinear opinion formation \modelName{} model, which is generalized from DeGroot's model, and accounts for both backfire effect and biased assimilation. Finally, we provide a comparison between our \modelName{} and the related biased opinion formation model on a simple example, to highlight their qualitative differences.

\subsection{Preliminaries and background}

\ptitle{Notation} Let $G = (V,E)$ denote a connected undirected network,
with $V = \left \{  1, ..., n\right \}$ the set of nodes,
and $E \in V \times V$ the set of $m = \left |  E\right |$ edges, where $(i,j) \in E$ iff $(j,i) \in E$.
When the network is weighted, $w_{ij}=w_{ji}$ represents the weight of edge $(i,j)$.
We use $N(i)$ to denote the set of neighbors of node $i$: $N(i) \triangleq \left \{ j \in V | (i,j ) \in E \right \}$.

In the considered models, opinions are real numbers within a fixed interval $[0,1]$ or $[-1,1]$,
depending on the model.
To discriminate between the two, we use $x$ to denote the opinions within $[0,1]$, and $y$ to denote the opinions that belong to $[-1,1]$.
All models we consider in this work can be defined as dynamical systems, where opinions are updated iteratively.
We use $x_i(t)$ (resp. $y_i(t)$) to denote the opinion of node $i$ at iteration (time) $t=0,1,2,\ldots$.
We further use $\mathbf{x}(t)$ and $\mathbf{y}(t)$ to denote the opinion vectors for the network at time $t$.
With $x_i$ (resp. $y_i$) we denote the opinion of node $i$ after convergences for $t\rightarrow\infty$ (if that limit exists), and $\mathbf{x}$ (resp. $\mathbf{y}$) to denote the corresponding vectors.

\ptitle{DeGroot's Model}
This model~\cite{degroot1974reaching} is an averaging opinion formation model, where the individual's opinion is determined by the average of her own opinion and that of her neighbors.
More specifically, 
it is updated as follows:
\begin{align}
x_i(t+1) = \frac{w_{ii} x_i(t) +\sum_{j \in N(i)} w_{ij}x_j(t) }{w_{ii} + \sum_{j \in N(i)}w_{ij}} \label{eq:degroots}
\end{align}
where $w_{ii}$ represents the extent to which the node values its own opinion, and $w_{ij}$ is the strength of the connection/friendship between node $i$ and $j$.
Iterative opinion updates will converge to a stationary state, where every node has the same opinion $x_i = x^*$~\cite{Jackson-book}. Therefore, the model always reaches consensus, and never polarizes.

\ptitle{Biased Opinion Formation}
The BOF model~\cite{dandekar2013biased} generalizes DeGroot's to incorporate \textit{biased assimilation}.
Given a weighted undirected graph $G = (V,E,w)$, every node $i \in V$ is assigned a bias parameter $b_i \geq 0$. Higher values of $b_i$ means that node $i$ is more biased. The opinion value $x_i (t) \in \left [  0,1\right ]$ 
is interpreted as the degree of support for opinion position $1$ (i.e., the highest possible opinion value), while $1-x_i(t)$ is the support for $0$. It is defined as
\begin{align*}
x_i (t+1) = \frac{w_{ii}x_i(t) + (x_i(t))^{b_i}s_i(t)}{w_{ii}+  (x_i(t))^{b_i}s_i(t) + (1 - x_i(t))^{b_i}(d_i - s_i(t))}
\end{align*}
where $s_i(t) \triangleq \sum_{j \in N(i)}w_{ij} x_j(t)$ is the weighted sum of $i$'s neighbouring opinions, and $d_i \triangleq \sum_{j \in N(i)}w_{ij}$ is the weighted degree of node $i$.
During the updating process, node $i$ weighs confirming and disconfirming evidence in a biased way: weighing the neighboring support for opinion $1$ by $(x_i(t))^{b_i}$, and that for opinion $0$ by $(1 - x_i(t))^{b_i}$.

\subsection{The BEBA model}
We now define the BEBA model, which is a generalization of DeGroot's model that incorporates both biased assimilation and backfire effect.
To capture these phenomena, we adapt DeGroot's model by dynamically setting the weights on the edges. 
Let $\mathbf{y}(t)$ denote the vector of opinions at time $t$, with $y_i(t) \in [-1,1]$.
Then, rather than using fixed weights as in DeGroot's model, we propose to let the weights be determined by the opinions as well. Specifically, for an edge $(i,j)\in E$ we define the edge weight $w_{ij}(t)$ at time $t$ as
$$w_{ij}(t) = \beta_i y_i(t)y_j(t) + 1.$$
The product $y_i(t)y_j(t)$ captures the degree of (dis)agreement between the opinions of node pair $(i,j)$. 
The parameter $\beta_i>0$ models the influence for $i$ that the (dis)agreement with node $j$ will have on the weight $w_{ij}(t)$: the larger, the stronger the biased assimilation and backfire effects.
We will refer to $\beta_i$ as the \emph{entrenchment parameter} of node $i$.

Given the weight $w_{ij}(t)$, the opinions in the BEBA model are updated as in DeGroot's model:
\begin{align}
y_i(t+1) = \frac{w_{ii}y_i(t) + \sum_{j \in N(i)} w_{ij}(t)y_j(t)}{w_{ii} + \sum_{j \in N(i)} w_{ij}(t)}
\label{eq:BABEmodel}
\end{align}
Note that when $\beta_i = 0$, BEBA's update rule is identical to that of DeGroot's (Eq.~(\ref{eq:degroots})) for unweighted networks. 
When $\beta_i \neq 0$, we discriminate two cases depending on $w_{ij}(t)$:
\begin{enumerate}[leftmargin=*,noitemsep,topsep=0pt,parsep=0pt,partopsep=0pt]
\item $w_{ij}(t) < 0$: This case models the backfire effect where $\beta_i y_i(t)y_j(t) < -1$. Since $\beta_i > 0$, $y_i(t)y_j(t)<0$, that is, nodes $i$ and $j$ hold opposing views. 
Multiplying $y_j(t)$ with this negative weight $w_{ij}(t)$ in the summation in the numerator leads to a contribution of the same sign as $y_i(t)$, while adding the negative weight to the denominator reduces it, inflating the resulting quotient. The combination of these two effects models the backfire effect.

\item $w_{ij}(t) > 0$: This case models biased assimilation, including two subcases:  
	\begin{enumerate}[noitemsep,topsep=0pt,parsep=0pt,partopsep=0pt]
    \item $0<\beta_i y_i(t)y_j(t)$: 
		Thus, node $i$ and $j$ have both positive or both negative opinions, resulting in an increased weight $w_{ij}(t)$.
 In this case node $i$ assimilates the opinion of neighbor $j$ more strongly if the extent of their agreement is stronger. 
	\item $-1 < \beta_i y_i(t)y_j(t) < 0$: Here 
	nodes $i$ and $j$ hold opposing but not too different opinions. 
	In this case, node $i$ critically evaluates the conflicting opinion of node $j$, but still assimilates it to a reduced extent.
	\end{enumerate}
\end{enumerate}

Note that the denominator in Eq.~(\ref{eq:BABEmodel}) can become $0$ resulting in a diverging opinion, or negative causing an unnatural opinion reversal.
We consider this situation to be beyond the model's validity region, and thus define the \modelName{} model as:
\begin{align*}
y_i(t+1) = \left\{\begin{matrix}
 \mathrm{sgn} (y_i(t)) \;\;\; \mathrm{if} \; w_{ii} + \sum_{j \in N(i)} w_{ij}(t)  \leq 0,\\
 \frac{w_{ii}y_i(t) + \sum_{j \in N(i)} w_{ij}(t)y_j(t)}{w_{ii} + \sum_{j \in N(i)} w_{ij}(t)}  \;\;\; \mathrm{otherwise}.
\end{matrix}\right. 
\end{align*}

Moreover, for a small denominator the resulting opinions may fall outside the range $[-1,1]$. To address this, we additionally clip negative values at $-1$ and positive values at $1$.

\subsection{Comparison of the \modelName{} and BOF models}
\label{sec:comparison}

There is a similarity between the BOF and our BEBA model, in that both alter the weights of the DeGroot's.
Consider a simple star graph of five nodes where node $1$ is in the center, and focus on one iteration of updating on node $1$. 
In this case, we can observe how the two models update the opinion of a single node, given the opinions of her neighborhood.

First, we deal with the fact that BOF model assumes only positive opinion values, while our model assumes opinions being both positive and negative. 
Note that the value range of opinions is important in both models, since the BOF model weights the opinion values, while our model exploits the disagreement in the sign.
To compare the models, we assume positive opinion values $x_i(t) \in \left [ 0, 1\right ]$ on all nodes in the graph, and use them to implement an update of the BOF model.
For our model, we transform opinions to the range $[-1,1]$ by setting $y_i(t) = 2x_i(t) - 1$.
Then we compute the value $y_1(t+1)$ as defined in BEBA, and rescale back. 

\begin{figure}[tp]
	\centering
	\includegraphics[width = 0.9\hsize]{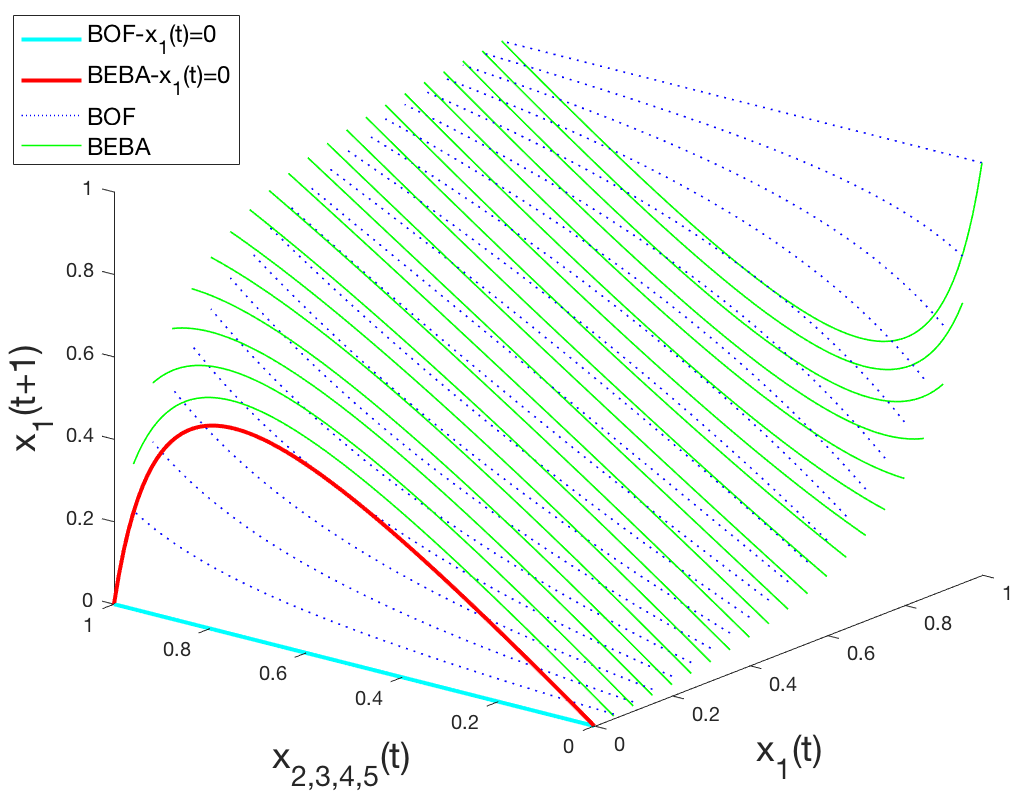}
	\caption{Opinion Formation on the Star Graph}
	\label{fig:star5of}
\end{figure}

In our experiment we assume $x_i(t)$ identical for all $i = 2,...,5$, and $x_i(t) \in [0,1]$ for all nodes.
We set $w_{11} = 1$ for both models, $b_1 = 1$ for BOF, and consider the values of $1$ and $2.5$ for $\beta_1$ in BEBA model. 
The opinion value $x_1(t+1)$ for both models, as a function of $x_{2,3,4,5}(t)$ and $x_1(t)$ is shown in Figure~\ref{fig:star5of}. 
The difference between the two models becomes clear when $x_1(t)$ takes extreme values (i.e., $0$ or $1$).

Figure~\ref{fig:star5Slice}(a) shows the curves for the two models when $x_1(t) = 0$. 
In BOF, the opinion $x_1(t+1)$ remains unchanged at value $0$. 
This is true regardless of the value of $b_1$. 
Thus, extreme nodes never change their opinions, even a little, even when they are not biased at all. 
However, according to the biased assimilation, unbiased individuals should be influenced by similar opinions, while even extreme nodes assimilate opinions that are close to their own. 
In contrast, our model better captures the biased assimilation in this case. 
In Figure~\ref{fig:star5Slice}(a), for $\beta_1 = 1$, which corresponds to a mildly biased node, the opinion of node $1$ can be moderated by that of her neighbors to different extents, while $x_1(t+1)$ never exceeds $0.5$. Therefore, extreme nodes are not stuck in the extremes.
\begin{figure}[tp]
\centering
\includegraphics[width = 1.0\hsize]{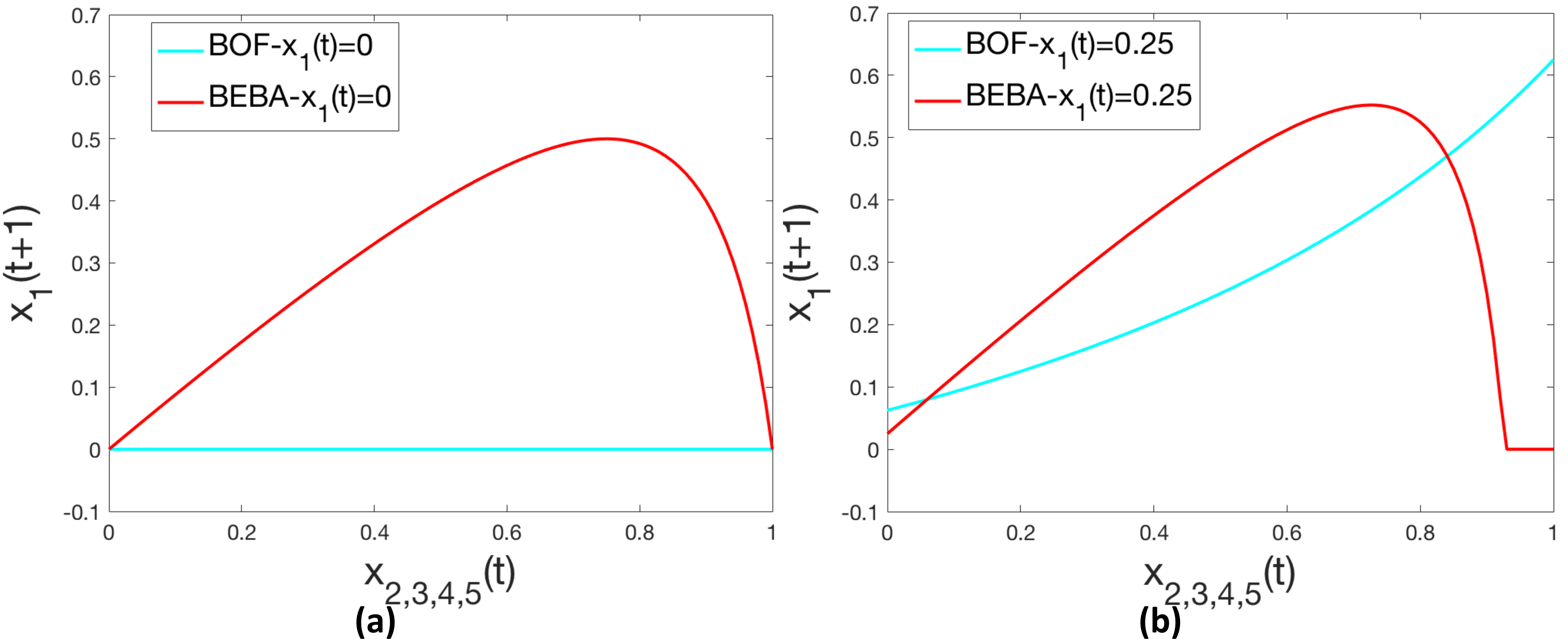}
\caption{$x_1(t+1)$ as a function of $x_{i}(t)$, (a) $\beta_1 = 1$, $b_1 = 1$, $x_1(t) = 0$; (b) $\beta_1 = 2.5$, $b_1 = 1$, $x_1(t) = 0.25$.}
\label{fig:star5Slice}
\end{figure}

To better understand the backfire effect, we increase $\beta_1$ to $2.5$, and set $x_1(t) = 0.25$ as shown in Figure~\ref{fig:star5Slice}(b).
We observe that when the disagreement between node 1 and her neighbors becomes large (i.e., $> 0.9$), $x_1(t+1)$ drops under $0.25$, until it becomes completely extreme with value $0$.

From the plots in Figure~\ref{fig:star5Slice} we also observe that for the different combinations of $\beta_1$ and $x_1(t)$,
there exists a value of the neighboring opinions that causes the largest change in $x_1(t+1)$. 
For example, when $\beta_1 = 1$ and $x_1(t) = 0$, neighboring opinion of around $0.75$ is the most influential as shown in Figure~\ref{fig:star5Slice}(a); for $\beta_1 = 2.5$ and $x_1(t) = 0.25$, opinion around $0.7$ is the most influential according to Figure~\ref{fig:star5Slice}(b).

\section{Theoretical Analysis} \label{sec:theoreticalAnalysis}
This section contains theoretical analysis of the BEBA model for two settings\footnote{Supplemental materials including theoretical proofs, datasets information, and more experimental results available in the Appendix.}. First we investigate the dynamics of opinions for a single agent in a fixed environment, and secondly we study the dynamics of polarization for all nodes in a connected social network.

\subsection{A single agent in a fixed environment} \label{sec:agentFixEnv}
Here we theoretically analyze the limit behavior of a single agent's opinion in an environment with a fixed opinion.
An analysis of this type has been done for the BOF model~\cite{dandekar2013biased}.
The setup is admittedly somewhat artificial but helps to gain a better understanding of the model.
It has been deemed realistic in cases where the fixed environment consists of the news media, billboards, etc.~\cite{dandekar2013biased}.
It also models the situation where the single agent is connected to a network
that is large enough such that adding it will not meaningfully affect the network.

For the agent $i$, we denote $y(t) \in \left [-1, 1\right ]$ its opinion at time $t$, $\beta > 0$ its entrenchment parameter, and $y$ its converged opinion (i.e., $\lim_{t\to\infty} y(t)$).
We assume the agent weighs its own opinion with $w_{ii} = w$. 
For simplicity, we only consider the situation where the environment contains one node,
but it should be noted that the analysis below can be easily generalized to several nodes.
Let $p \in [ -1,1]$ be the fixed environmental opinion.
Then, according to \modelName{}, the agent updates its opinion as follows:
\begin{align*} 
y(t+1) =\left\{\begin{matrix}
 \mathrm{sgn} (y(t))& \mathrm{if} \; w + \beta p y(t) + 1  \leq 0,\\ 
\frac{wy(t) + \beta p^2 y(t) + p}{w + \beta p y(t) + 1}&  \mathrm{otherwise}.
\end{matrix}\right.
\end{align*}

Before stating a theorem that quantitatively characterizes the limit $y$, we consider the behavior.
[Case 1:] For sufficiently small entrenchment $\beta$ (i.e., not biased), the fixed environment's opinion $p$ will be sufficiently attracting such that $y=p$ regardless of $y(t)$.
The same is true when $p=0$: the neutral opinion is never polarizing and thus always attracting.
[Case 2:] On the other hand, for sufficiently large entrenchment $\beta$ (i.e., biased),
the limit $y$ will depend on the similarity of initial opinion $y(t)$ with the environment's opinion $p$:
[Case 2a:] if $y(t)$ is similar to $p$, $p$ should have an attracting effect on $y(t)$ such that its limit $y=p$;
[Case 2b:] if $y(t)$ is very different from $p$, however,
the backfire effect will cause the agent's opinion to diverge from $p$, such that $y=\mathrm{sgn}(y(t))$.
[Case 2c:] Between Case 2a and Case 2b, there will be a `sweet spot' where $y(t)$ is neither sufficiently similar to $p$ for $y(t)$ to converge to $p$,
nor sufficiently different for it to diverge to $\mathrm{sgn}(y(t))$.
This is an unstable equilibrium where $y(t)$ remains constant through time, i.e. $y=y(t)$.

This intuition is formalized in the following theorem.
For conciseness and transparency, we state it for the situation where $p\leq 0$.
It is trivial to adapt the theorem for $p\geq 0$.
\begin{theorem} Depending on the value of $\beta$ relative to $p$:\label{thm:singleagentfixedenvironment}
\begin{description}[noitemsep,topsep=0pt,parsep=0pt,partopsep=0pt]
\item[Case 1:] When $p=0$ or $\beta <-1/p$, the agent's opinion always converges to $p$, i.e., $y=p$.
\item[Case 2:] When $p<0$ and $\beta \geq -1/p$, there are three possibilities depending on how similar $y(t)$ is to $p$.
(This situation is illustrated in Figure~\ref{fig:Theorem1EG}.)
\begin{description}[noitemsep,topsep=0pt,parsep=0pt,partopsep=0pt]
\item[a:] If $y(t)<-\frac{1}{\beta p}$, $y(t)$ will be sufficiently attracted to $p$ such that $y=p$.
\item[b:] If $y(t)>-\frac{1}{\beta p}$, $y(t)$ will diverge away from $p$ such that $y=\mathrm{sgn}(y(t))=1$.
\item[c:] If $y(t)=-\frac{1}{\beta p}$, $y(t)$ will remain constant through time, such that $y=-\frac{1}{\beta p}$.
\end{description}
\end{description}
\end{theorem}

\begin{figure}[tp]
\centering
\includegraphics[width = 0.7\hsize]{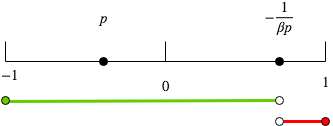}
\caption{Graphical illustration of Case 2 from Theorem~\ref{thm:singleagentfixedenvironment} (i.e. $p<0$ and $\beta \geq -1/p$).
[Case 2a:] For values of $y(t)$ in the green range, $y(t)$ will converge to $y=p$.
[Case 2b:] For values of $y(t)$ in the red range, $y(t)$ will diverge to $y=1$.
[Case 2c:] For $y(t)=-\frac{1}{\beta p}$, $y(t)$ will not change such that $y=-\frac{1}{\beta p}$.}
\label{fig:Theorem1EG}
\end{figure}

Theorem~\ref{thm:singleagentfixedenvironment} already suggests that opinions under the \modelName{} model
evolve to one of three possible states: consensus (Case 1 and Case 2a),
polarization (Case 2b), and an unstable state of persistent disagreement (Case 2c).

\subsection{Polarization and consensus for general networks and initial opinions}
Here we extend from the single agent to a group of individuals that can update their opinions at any time step $t$. The dynamics of polarization are investigated theoretically with respect to different values of the entrenchment parameter. It was argued by the authors of the BOF model that homophily alone, without biased assimilation was not sufficient for polarization~\cite{dandekar2013biased}. In our BEBA model, the backfire effect and biased assimilation, without homophily, are sufficient to lead to polarization or consensus, depending on the parameters and the initial opinions.
The theorem below makes this clear, by providing easy-to-realize sufficient conditions for polarization or consensus to occur.
\begin{theorem}\label{theo:polarization}
Let $G = (V, E)$ be any connected unweighted undirected network. For all $i \in V$, $y_i(t) \in (-1,0)\cup (0,1)$ is the opinion of node $i$ at time $t$, let $w_{ii} = 1$ and $\beta_i = \beta >0$ for all $i \in V$. Denote $\mathbf{y}(t)$ the opinion vector of $G$ at time $t$, $ \left | \mathbf{y}(t) \right | $ is the vector with the absolute values of all opinions, and $\mathrm{min}(\mathbf{y}(t))$ is the minimum element in $\mathbf{y}(t)$. Then,
\begin{enumerate}[noitemsep,topsep=0pt,parsep=0pt,partopsep=0pt]
\item Polarization: If $\beta > \frac{1}{\left [ \mathrm{min} \left ( \left | \mathbf{y}(0) \right | \right ) \right ]^2}$, $\forall i \in V$, $ \left | y_i\right |= 1$.
\item Consensus: If $\beta < \frac{1}{\left [ \mathrm{max} \left ( \left | \mathbf{y}(0) \right | \right ) \right ]^2}$, there exists a unique $y^* \in [-\mathrm{max} \left ( \left | \mathbf{y}(0) \right | \right ) ,\mathrm{max} \left ( \left | \mathbf{y}(0) \right | \right ) ]$ such that $ y_i = y^*$, $\forall i \in V$.
\end{enumerate}
\end{theorem}

A special case of particular theoretical interest is when $\mathrm{min} \left ( \left | \mathbf{y}(0) \right | \right )  = \mathrm{max} \left ( \left | \mathbf{y}(0) \right | \right )$.
Then there are only two different initial opinions in the network,
with the same absolute value but opposite signs (i.e. they could represent 'for' and 'against' an issue of interest).
In this case, the sufficient conditions also become necessary conditions,
and a borderline situation emerges to which we refer as \emph{persistent disagreement}.
It can be proved concisely by relying on Theorem~\ref{theo:polarization}, and thus we state it as a Corollary:
\begin{corollary} \label{corollary:polarization}
Let $G = (V_1, V_2, E)$ be any connected unweighted undirected network. For all $i \in V = V_1 \cup V_2$, let $w_{ii} = 1$ and $\beta_i = \beta > 0$. Assume for all $i \in V_1$, $y_i(0) = y_0$, where $0 < y_0 < 1$; while for all $i \in V_2$, $y_i(0) = - y_0$. Then,
\begin{enumerate}[noitemsep,topsep=0pt,parsep=0pt,partopsep=0pt]
\item Polarization: If $\beta > \frac{1}{y_0^2}$, $\forall i \in V_1\cup V_2$, $ \left | y_i \right | = 1$.
\item Persistent disagreement: If $\beta = \frac{1}{y_0^2}$ (i.e., when $w_{ij} = 0$ if $i \in V_1$ and $j \in V_2$), $\forall i \in V_1$, $y_i(t') = y_0$ for all $t' \geq 0$, and $\forall i \in V_2$, $y_i(t') = -y_0$ for all $t' \geq 0$.
\item Consensus: If $\beta < \frac{1}{y_0^2}$, then there exists a unique $y^* \in (-y_0, y_0)$ such that $\forall i \in V$, $ y_i = y^*$.
\end{enumerate}
\end{corollary}

Intriguingly, these conditions in the Theorem and Corollary are independent of the network structure and depend only on the entrenchment parameter $\beta$ and the opinion vector at time $0$.
Yet, it should be noted that the value of the consensus and the eventual polarized state do depend on the network structure.
Moreover, the network structure, and the distribution of the opinions over it,
do determine whether polarization or consensus will arise when neither of the sufficient conditions of Theorem~\ref{theo:polarization} are satisfied.
These claims are confirmed in experiments in the next section.

\section{Experimental Analysis} 
\label{sec:experiments}
In Section~\ref{sec:theoreticalAnalysis} we provided sufficient conditions for our model to reach consensus or polarization. 
In this section we perform an experimental analysis of how these two phenomena manifest themselves on real and synthetic networks. Our goal is to answer the following questions:
\begin{itemize}[noitemsep,topsep=0pt,parsep=0pt,partopsep=0pt]
\item In the case that the network reaches consensus, what is the value of the consensus opinion, and how does the network structure, $\beta$, and the initial opinion vector affect this value? 
\item In the case that the opinions polarize, what is the state of the polarization and how is it affected by the initial opinions, $\beta$, and the structure of the networks? 
\end{itemize}

We use both real-world and synthetic networks in our experiments. The real datasets include Zachary's Karate Club network~\cite{zachary1977information} and six Twitter networks with given opinions for different events ranging from political elections to sports  activities~\cite{zarezade2017cheshire,de2016learning}. See the Appendix for details.
The synthetic networks are:
\begin{itemize}[noitemsep,topsep=0pt,parsep=0pt,partopsep=0pt]
\item Erd\H{o}s-R\'enyi (ER) networks $G(n,\rho)$ have binomial degree distributions, where $\rho$ is the edge connection probability between nodes~\cite{bollobas2001random}.
\item Watts-Strogatz (WS) networks $G(n,K, 1)$ have the small world property of with $K$ being the average degree, and we fix the rewiring probability to be $1$ (i.e., random graph), thus only refer to $K$~\cite{watts1998collective}.
\item Barab\'asi-Albert (BA) networks $G(n, M_0,M)$ are scale-free, where $M_0$ is the number of initial nodes and $M$ the number of nodes that a new node is connected to~\cite{albert2002statistical}.
\end{itemize}

\subsection{The influence of the entrenchment $\beta$}
From Theorem~\ref{theo:polarization}, we know the stationary opinion vector $\mathbf{y}$ of our model polarizes when $\beta > \frac{1}{[\mathrm{min} (\left | \mathbf{y} (0)\right |)]^2}$, and reaches consensus when $\beta < \frac{1}{[\mathrm{max} (\left | \mathbf{y} (0)\right |)]^2}$. 
However, these limits are far away from each other and polarization may occur at much lower values of $\beta$ in practice, similarly consensus for higher $\beta$. We now take the Karate network as an example and examine the relation between $\beta$ and polarization experimentally using random initial opinion vectors.

\begin{figure}[tp]
	\centering
	\includegraphics[width = 1.0\hsize]{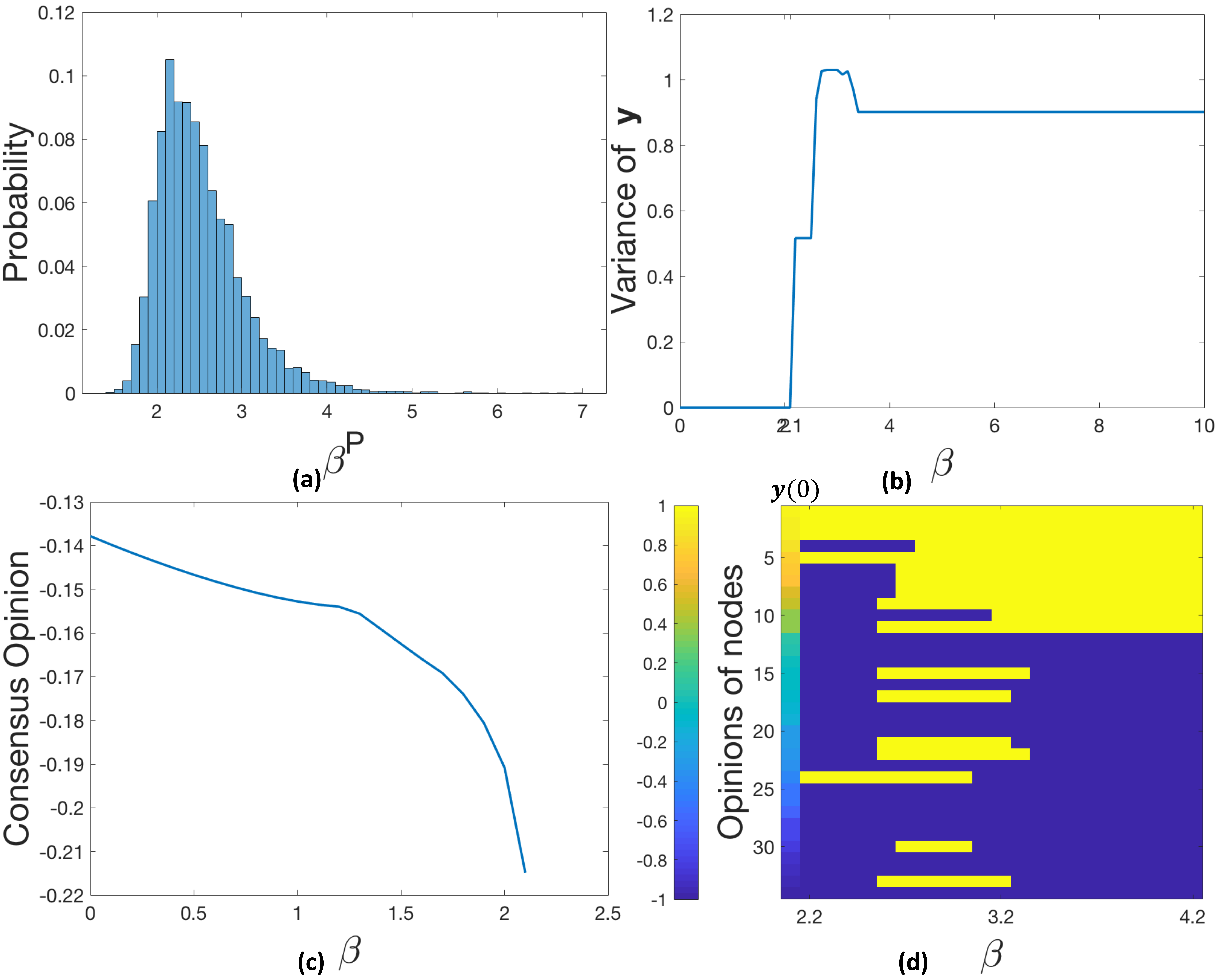}
	\caption{For the Karate network: (a) the distribution of $\beta^P$ (i.e., the smallest $\beta$ that results in polarization) for $10000$ random opinion vectors (uniform on $\left [ -1,1\right ]$); for one opinion vector, (b) the variance of all converged $\mathbf{y}$ as $\beta$ increases from $0$ to $10$; (c) consensus opinion values for $\beta \in [0,2.1]$; (d) final opinions for each of the nodes.}
	\label{fig:E1Result}
\end{figure}

Let $\beta^P$ denote the threshold between consensus and polarization for any pair of network and opinion vector. 
Figure~\ref{fig:E1Result}(a) shows the distribution of the empirical $\beta^P$ values for 10000 different random opinion vectors, where $y_i(0)$ is uniform within $[-1, 1]$.
We observe that the threshold for polarization is much smaller than the theoretical value, which should be lager than $10^4$. 
However, the empirical value of $\beta^P$ is below $5$ for most of the $\mathbf{y}(0)$, and never exceeds $7$. 

In Figure~\ref{fig:E1Result}(b), the variance of the stationary opinion vector is plotted as a function of $\beta$, for one of the opinion vectors. When there is consensus the variance is zero, while when the variance is greater than zero, polarization is obtained (i.e., different variances correspond to different polarized states). We observe that as $\beta$ increases, the opinion vector converges from consensus to polarized states. Empirically, no persistent disagreement is achieved. For this $\mathbf{y}(0)$, polarization is shown if $\beta > 2.1$ such that $\beta^P = 2.1$. 

When reaching consensus, Figure~\ref{fig:E1Result}(c) shows that the consensus value becomes less neutral as $\beta$ increases. This is true for $78.74\%$ of the 10000 vectors on Karate network. 
Meanwhile, different $\beta$s do not necessarily result in the same polarized state (see Figure~\ref{fig:E1Result}(d)). The heatmap shows different polarized states for different values of $\beta$ for this $\mathbf{y}(0)$.

\subsection{The influence of the opinion vector $\mathbf{y}(0)$}
In this experiment, we investigate the influence $\mathbf{y}(0)$ on the consensus opinion value and the mean polarized opinion.
Figure~\ref{fig:E2Result} shows that the consensus value and the mean polarized opinion are strongly correlated to the mean of $\mathbf{y}(0)$. 
Meanwhile, Figure~\ref{fig:E2Result}(b) shows that in the case of polarization, opinion vectors with similar initial means may result in quite different polarized states because the placements of the opinions on nodes differ. Also, $\mathbf{y}(0)$ with different means may result in similar polarization (i.e., mean polarized opinion). 

\begin{figure}[tp]
\centering
\includegraphics[width = 1.0\hsize]{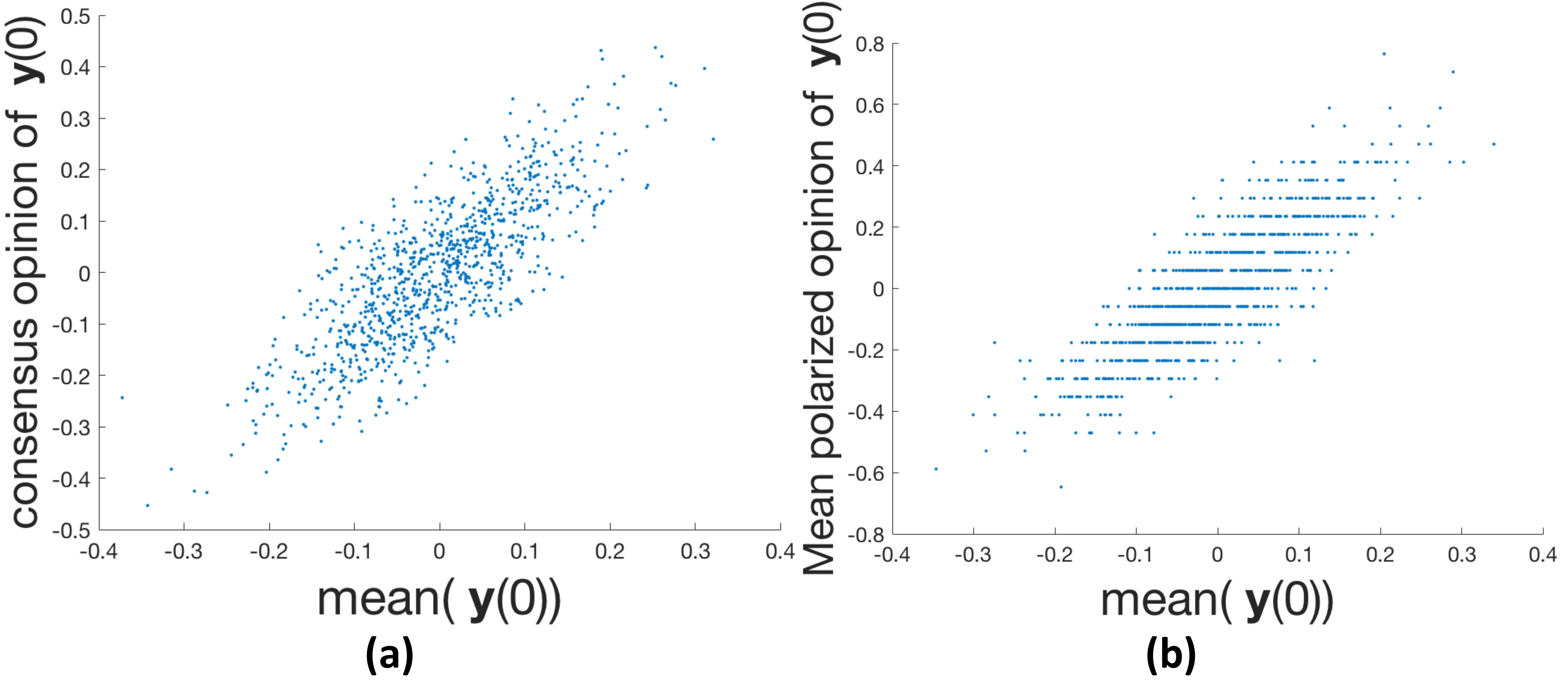}
\caption{For $1000$ random $\mathbf{y}(0)$ on Karate network: (a) consensus opinion when $\beta = 1$; (b) mean polarized opinion when $\beta = 10$.}
\label{fig:E2Result}
\end{figure}

\begin{figure}[tp]
	\centering
	\includegraphics[width = 1.0\hsize]{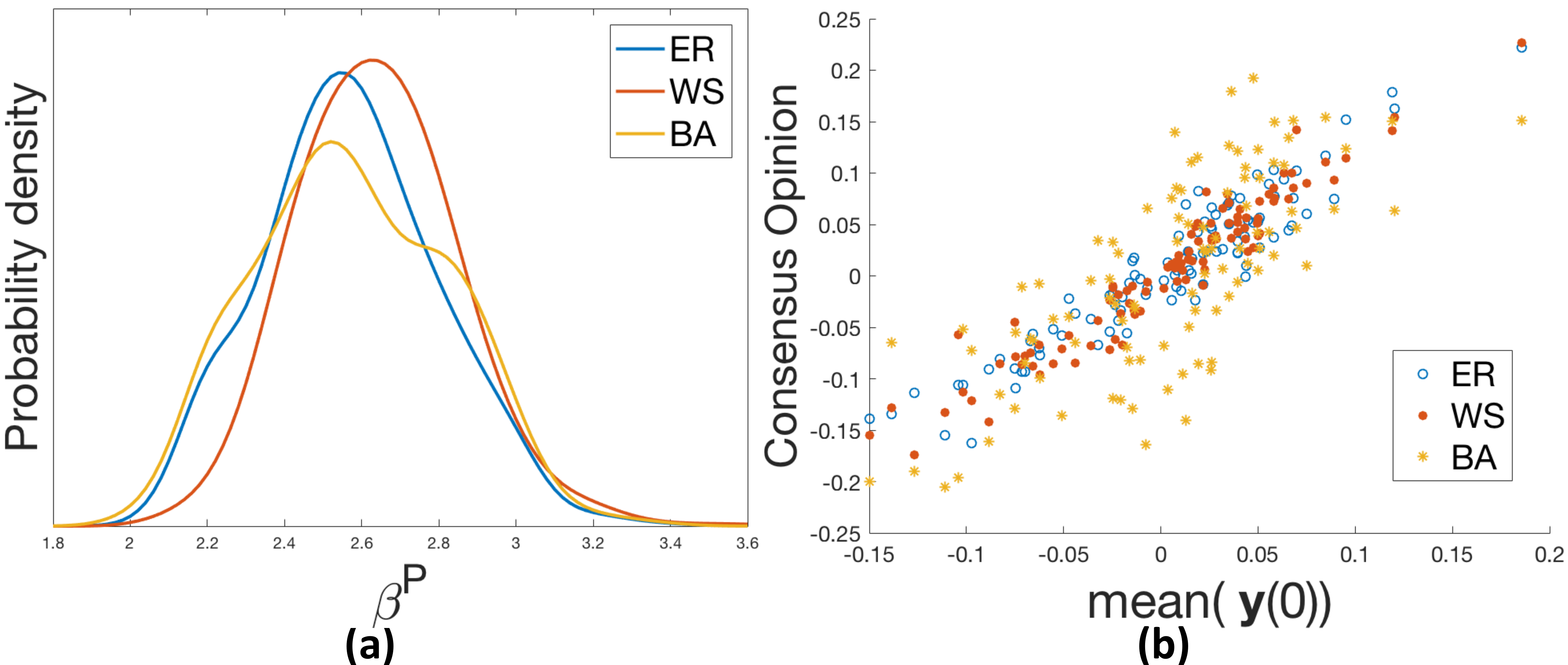}
	\caption{Based on one ER ($n = 100, \rho = 0.0606$), one WS ($n = 100, K = 3$), and one BA network ($n = 100, M_0 = 4, M = 3$): (a) distribution of $\beta^P$ for $1000$ random opinion vectors; (b) for $100$ opinion vectors, mean $\mathbf{y}(0)$ vs. the consensus value ($\beta = 1$).}
	\label{fig:E31Result}
\end{figure}

Then we analyze two real datasets Tw:Club (i.e., Barcelona getting the first place in La-liga 2016) and Tw: Sport (Champions League final in 2015 between Juventus and Real Madrid), which have the same network but different initial opinion vectors. It is found that the $\beta^P$ is 11.7 for Tw:Club and 3.3 for Tw:Sport, which indicates the Champions League final gets polarized more easily than the other event.

\subsection{The influence of the network topology $G$}
In this experiment, we study how the topology affects the $\beta^P$ for the same (set of) $\mathbf{y}(0)$, as well as the stationary opinion vectors of our model.
To this end, we generated networks with the three random network models, with the same number of nodes, intialized with the same opinion vectors. 

We observe that for networks with the same number of nodes and similar numbers of edges, different network properties result in different dynamics of polarization. Figure~\ref{fig:E31Result}(a) shows that for the same set of $\mathbf{y}(0)$, the distributions of the $\beta^P$ value for the three models. It shows that the $\beta^P$ has a larger mean in the WS model, indicating networks with this structure may be more robust against polarization.
We also observe the standard deviation of the $\beta^P$ values for the BA distribution is larger, which appears to be due to 'hub' nodes, whose opinions strongly affect the value of $\beta^P$.

Figure~\ref{fig:E31Result}(b) plots the consensus values reached by a set of 100 random opinion vectors on the three networks. The shapes of scatter plots become increasingly compact from the BA model, the ER model, to the WS model, corroborating the larger variance in the opinion dynamics on BA networks.

The parameters in each model also affect the dynamics, see the supplement in the Appendix. 
For example, when the edge probability $\rho$ in the ER model increases from a small number, which guarantees a connected network, to $1$, $\beta^P$ varies less for ER models with similar $\rho$. 
The experimental results are similar for the consensus value, and the polarized opinion. 
Not only the number of edges has an influence on the dynamics of polarization, but also the placement of the edges.

\subsection{Real-world dataset analysis}
\begin{table}[tp]
\centering
\caption{$\beta^P$ for real-world twitter datasets}
\label{tab:realBetaP}
\begin{tabular}{cc|cc|cc}
\hline
Network  & $\beta^P$ & Network  & $\beta^P$ & Network  & $\beta^P$ \\ \hline
Tw:GoT & 2.9 & Tw:Club  & 3.3  & Tw:US    & 4.9 \\
Tw:UK    & 7.5  & Tw:Delhi & 7.7 & Tw:Sport & 11.7 
\end{tabular}
\end{table}
Based on the six real-world twitter datasets~\cite{zarezade2017cheshire,de2016learning}, we investigate how easily each event gets polarized opinions, namely the value of $\beta^P$. It is shown in Table~\ref{tab:realBetaP} that political events are apparently less likely to polarize, except the US one. While the sport or TV events are more likely get polarized, except when people had to bet instead of supporting (i.e., Tw:Club).

\section{Conclusion and Future Work} \label{sec:conclusion}

Modeling how opinions evolve when individuals interact in social networks is an important computational social science challenge that has received renewed attention recently. The availability of realistic models of this type may have substantial real-life impact on a variety of applications, from political campaigns design, to conflict prevention and mitigation.

A large number of models have been proposed in the literature. To the best of our knowledge, however, none of them model the so-called Backfire Effect: the fact that individuals, when exposed to a strongly opposing view, will not be moderated, but rather become more entrenched in their opinion.

Here we proposed the \modelName{} model, which models both Biased Assimilation and Backfire Effect. It is governed by one parameter (which can vary over the individuals), called the entrenchment parameter, determining the strength of both. The \modelName{} model naturally generates different behaviors: from convergence to a consensus, to polarization.

Theoretical and empirical analyses demonstrate that the resulting model is not only realistic, its behavior also provides an interesting view on the interplay between network structure, the entrenchment parameter, and the opinions.

These properties make the \modelName{} model a useful tool for simulating the effect of interventions, such as editing the network (e.g. by facilitating communication between particular pairs of individuals), altering the initial opinions (e.g. through targeted information campaigns), or affecting the entrenchment of particular individuals (e.g. through education).

\ptitle{Acknowledgements} This work was supported by
the ERC under the
EU's Seventh Framework Programme (FP7/2007-2013) / ERC
Grant Agr. no.\ 615517,
FWO (project no.\ G091017N, G0F9816N),
the EU's Horizon 2020 research and innovation programme and the FWO under the Marie Sklodowska-Curie Grant Agr. no.\ 665501.

\bibliographystyle{named}
\bibliography{main}

\appendix

\section{Proof of Theorem~\ref{thm:singleagentfixedenvironment}} \label{proof:theorem1}
\subsection{Only one node in the environment}
Recall that there is one node with a fixed opinion $p \in [-1,1]$ in the environment. The opinion of the agent is updated as mentioned in Section~\ref{sec:theoreticalAnalysis}.
\begin{lemma}
If $w + \beta p y(t) + 1  \leq 0$, the opinion of the agent stays at $\mathrm{sgn}(y(t))$ for all $t'>t$.
\end{lemma}
\begin{proof}
As shown in the updating rule that when $w + \beta p y(t) + 1 \leq 0$, $y(t+1)= \mathrm{sgn}(y(t))$.
$w + \beta p y(t) + 1 \leq 0$ is equivalent to $\beta p y(t) \leq -w -1 < 0$.
Knowing that $\left | y(t+1)\right |  = 1 \geq \left | y(t)\right |$,
$$\beta p y(t+1) \leq -w -1$$
Therefore, $y(t') = \mathrm{sgn}(y(t+1)) = \mathrm{sgn}(y(t))$ for all $t'>t$.
\end{proof}

\begin{lemma}
If $w + \beta p y(t) + 1  > 0$, there exist two fixed points where $y(t+1) = y(t)$: $p$ and $-\frac{1}{\beta p}$. 
$p$ is attracting while $-\frac{1}{\beta p}$ is repelling.
\end{lemma}
\begin{proof}
The converged opinion $y$ of the agent should satisfy
$$f(y) = \frac{wy + \beta p^2y + p}{  w + \beta p y + 1 }$$
\begin{align}
f(y) - y = \frac{-\beta p y^2 + (\beta p^2 - 1)y + p}{w +\beta p y  +1} = \frac{u(y)}{v(y)}  = 0\label{al:fyminusy}
\end{align}
where 
\begin{align*}
u(y) & = -\beta p y^2 + (\beta p^2 - 1)y + p \\
v(y) & = \beta p y + w +1
\end{align*}
By solving $u(y) = 0$, which is equivalent to $f(y) - y = 0$ since $u(y) > 0$, the two fixed points of $f(y)$ are:
$p$ and $-\frac{1}{\beta p}$.

Next, we prove that $p$ is attracting and $-\frac{1}{\beta p}$ is repelling. 
$$f'(y) = \frac{w(w + \beta p^2)}{(w + \beta p y +1)^2} \geq 0$$
$\left | f'(y) \right | = f'(y)$, then $f'(p) = \frac{w(w + \beta p^2)}{(w + \beta p^2 +1)^2} < 1$, thus attracting; while $f'( y-\frac{1}{\beta p}) = \frac{w(w + \beta p^2)}{w^2} > 1$, thus repelling.
\end{proof}

\begin{lemma}
If $w + \beta p y(t) + 1  > 0$ and $py(t) \geq 0$, $y = p$.
\end{lemma}
\begin{proof}
If $p = 0$, $y(t+1) = \frac{w}{w+1}y(t)$, as the iteration goes, $\lim_{t \rightarrow \infty} y(t) = 0$;

If $py(t) > 0$, e.g., they are both positive
\begin{itemize}
\item when $0<y(t)<p$, $y(t+1) - y(t) = \frac{u(y(t))}{v(y(t))} > 0$, thus $y(t+1) > y(t)$, the agent's opinion increases until it reaches $p$;
\item when $p<y(t)< 1$, $y(t+1) - y(t) < 0$, the agent's opinion decreases to $p$.
\end{itemize}
\end{proof}
\begin{lemma}
If $w + \beta p y(t) + 1  > 0$ and $py(t) < 0$, 
\begin{enumerate}
\item If $\left| \frac{1}{\beta p} \right| > 1$, $\lim_{t \rightarrow \infty} y(t) = y^e$.
\item If $\left| \frac{1}{\beta p} \right | \leq 1$,
\begin{enumerate}
\item If $\left | y(t) \right | <  \left | \frac{1}{\beta p} \right |$, $y = p$.
\item If $y(t) = -\frac{1}{\beta p}$, $y(t')=-\frac{1}{\beta p}$ for all $t' \geq t$.
\item If $ \left | \frac{1}{\beta p}\right | < \left | y(t) \right | \leq 1$, $y = \mathrm{sgn}(y(t))$.
\end{enumerate}
\end{enumerate}
\end{lemma}
\begin{proof}
Assume $y(t) \in (0,1]$ and $p \in (-1,0)$, 
\begin{itemize}
\item if $\left| \frac{1}{\beta p} \right| > 1$, all $y(t) \in (0,1] < -\frac{1}{\beta p}$, $y(t)$ is attracted to $p$ as the updating goes; 
\item if $\left| \frac{1}{\beta p}\right| = 1$, $y(t)$ is repelled by the extreme point and goes to the attracting one unless it starts with $-\frac{1}{\beta p}$ at time $t$; 
\item if $\left| \frac{1}{\beta p} \right| < 1$, when $0 < y(t) < -\frac{1}{\beta p}$, $y(t+1) - y(t) = \frac{u(y(t))}{v(y(t))} < 0$, $y(t+1) < y(t)$, the agent's opinion decreases to $p$; when $y(t) = -\frac{1}{\beta p} $, $y(t)$ stays there; when $y(t) > -\frac{1}{\beta p} $, $y(t+1) > y(t)$, the agent's opinion increases to the extreme value on its side.
\end{itemize}
\end{proof}

\subsection{A group of nodes in the environment} \label{proof:theorem12}
Assume there is a set of $m$ neighbour with different fixed opinions, $\mathbf{p} = (p_1, p_2, ..., p_m)$, $m >1$. We denote
\begin{itemize}
\item $q = \sum_j p_j^2$ the sum of the squares of the fixed opinions.
\item $s = \sum_j p_j$ the sum of the fixed opinions.
\item $m = \sum_j 1$ the number of nodes in the environment.
\end{itemize}
\begin{lemma}
$mq - s^2 \geq 0$, which is $m\sum_j p_j^2 \geq (\sum_j p_j)^2$.
\end{lemma}
\begin{proof}
$$ m\sum_j p_j^2 - (\sum_j x_j)^2  = \frac{1}{2} \sum_i \sum_j (p_i - p_j)^2 \geq 0$$
\end{proof}

The agent's opinion is updated by 
\begin{align}
y(t+1) = &\left\{\begin{matrix}
 \mathrm{sgn} (y(t)) & \mathrm{if} \; w + \beta s y(t) + m  \leq 0,\\ 
  \frac{wy(t) + \beta q y(t) + s}{  w + \beta s y(t) + m } &  \mathrm{otherwise}.
\end{matrix}\right. \label{al:a2fy} 
\end{align}

\begin{lemma}
If $w + \beta s y(t) + m  > 0$, there exist two fixed points where $y(t+1) = y(t)$: 
\begin{align*}
y^a = \frac{\beta q - m + \sqrt{\Delta}}{2\beta s} \;\;\;\;\; y^r = \frac{\beta q - m - \sqrt{\Delta}}{2\beta s}
\end{align*}
where $ \Delta = (\beta q - m)^2 + 4\beta s^2$. $y^a$ is attracting while $y^r$ is repelling.
\end{lemma}

\begin{proof}
The function is $f(y) = \frac{wy + \beta q y + s}{  w + \beta s y + m }$. The two fixed points satisfy $f(y) = y$. $\left | f'(y) \right | = f'(y)$ since
\begin{align*}
f'(y) = &\frac{(w + \beta q)(w + m) - \beta s^2}{(\beta s y + w + m)^2} \\
 =&\frac{w(w+m) + \beta q w + \beta (qm - s^2)}{(\beta s y + w + m)^2} > 0
\end{align*} 

For $y^a = \frac{\beta q - m + \sqrt{\Delta}}{2\beta s}$, $f'(y^a) < 1$ because
\begin{align*}
f'(y^a) - 1 
= &-\frac{1}{2}\frac{(m-\beta q)^2 + 4 \beta s^2 + (2w + m + \beta q)\sqrt{\Delta}}{(\beta s y^a + w + m)^2} \\
< &0
\end{align*}

For $y^r = \frac{\beta q - m - \sqrt{\Delta}}{2\beta s}$, $f'(y^r) > 1$ because
\begin{align*}
f'(y^r) - 1=  &-\frac{1}{2}\frac{(m-\beta q)^2 + 4 \beta s^2 - (2w + m + \beta q)\sqrt{\Delta}}{(\beta s y^r + w + m)^2} \\
= &-\frac{1}{2} \frac{A}{B}
\end{align*}
$ \frac{A}{B} < 0$ since $B > 0$ and it can be proved as below that $A<0$.
\begin{align*}
&\left [  (m-\beta q)^2 + 4 \beta s^2\right ]^2 - \left [  (2w + m + \beta q)\sqrt{\Delta}\right ]^2\\
= &4 \left [   (m-\beta q)^2 + 4 \beta s^2 \right ] \left [  \beta (s^2 - qm) - w(m + w + \beta q)\right ] \\
< &0
\end{align*}
Therefore, $y^a$ is attracting and $y^r$ is repelling.
\end{proof}

\section{Proof of Theorem~\ref{theo:polarization}} \label{proof:theorem2}
Recall that $y_i(t) \in (-1,0) \cup (0,1)$. Given any opinion vector $\mathbf{y}(0)$ of a given connected network $G = (V,E)$, the opinions can be divided into two groups $V_1$ and $V_2$ at any time $t$: a) $\forall i \in V_1$, $y_i(t) > 0$; b)$\forall i \in V_2$, $y_i(t) < 0$, and $V = V_1 \cup V_2$. Denote $n^s_i(t)$ the number of node $i$'s neighbors node that are in the \textbf{s}ame group with $i$ at time $t$, and $n^d_i(t)$ the number of neighbors in the \textbf{d}ifferent group. Specifically, they are denoted as
\begin{align*}
n^s_i(t) = &|N(i)^s|, N(i)^s = \left \{ j | j \in N(i),\; \mathrm{and} \;y_i(t) y_j(t) > 0 \right \}\\
n^d_i(t) = &|N(i)^d|, N(i)^d = \left \{ k | k \in N(i),\; \mathrm{and} \;y_i(t) y_k(t) < 0 \right \}
\end{align*}

\begin{lemma}
For node $i \in V $ fix $\beta_i = \beta > 0$, if $\beta > \frac{1}{\left [ \mathrm{min} \left ( \left | \mathbf{y}(0) \right | \right ) \right ]^2}$, $\lim_{t \rightarrow \infty} \left | y_i(t) \right |= 1$.
\end{lemma}
\begin{proof}
For node $i \in V$, the opinion is updated with BEBA.
If $\gamma = 1+\sum_{j \in N(i)}w_{ij} \leq 0$, $y_i(t+1)$ reaches the extreme value in one iteration due to strong backfire effect. 

While when $\gamma > 0$, for any $t>0$, $y_i(t+1) $ is updated as
\begin{align}
y_i(t) \frac{1 + \sum_{j \in N(i)^s}w_{ij}\frac{y_j(t)}{y_i(t)} + \sum_{k \in N(i)^d}w_{ik}\frac{y_k(t)}{y_i(t)} }{1 + \sum_{j \in N(i)^s}w_{ij} + \sum_{k \in N(i)^d}w_{ik}}  = y_i(t) \frac{C}{D}\label{al:updatePosDemon}
\end{align}

When $\beta > \frac{1}{\left [ \mathrm{min} \left ( \left | \mathbf{y}(t) \right | \right ) \right ]^2}$, for all $k \in N(i)^d$, $w_{ik} = \beta y_i(t)y_k(t) + 1 < 0$. The sums in Equation~(\ref{al:updatePosDemon}) satisfy: $\sum_{j \in N(i)^s}w_{ij}\frac{y_j(t)}{y_i(t)}$, $\sum_{j \in N(i)^s}w_{ij}$, $\sum_{k \in N(i)^d}w_{ik}\frac{y_k(t)}{y_i(t)}>0$, and $\sum_{k \in N(i)^d}w_{ik}<0$. 

Now we focus on the node that has the most moderate opinion, namely the node with absolute value of opinion $\mathrm{min} \left | \mathbf{y}(t) \right |$ at each time step, starting from time $0$. Knowing $C, D > 0$,
\begin{align}
C - D =   \sum_{j \in N(i)^s}w_{ij}(\frac{y_j(t)}{y_i(t)}-1) + \sum_{k \in N(i)^d}w_{ik}(\frac{y_k(t)}{y_i(t)}-1) \label{al:diffCD}
\end{align}
Since $y_i(t)$ has the smallest absolute opinion value, for any $j \in N(i)^s$, $\frac{y_j(t)}{y_i(t)} \geq 1$, thus $C>D$, $\frac{C}{D}> 1$, and $\left | y_i(t+1) \right | > \left | y_i(t)\right |$.

After every iteration from time $t$ to $t+1$, the opinion of the most moderate node becomes more extreme, until it reaches the absolute value of $1$, thus for any $i \in V$, $\lim_{t \rightarrow \infty} \left | y_i(t) \right |= 1$.
\end{proof}

\begin{lemma}
For node $i \in V $, if $\beta < \frac{1}{\left [ \mathrm{max} \left ( \left | \mathbf{y}(0) \right | \right ) \right ]^2}$, there exists a unique $y^* \in [-\mathrm{max} \left ( \left | \mathbf{y}(0) \right | \right ) ,\mathrm{max} \left ( \left | \mathbf{y}(0) \right | \right ) ]$ such that $\lim_{t \rightarrow \infty} y_i(t) = y^*$ for all $i \in V$.
\end{lemma}
\begin{proof}
When $\beta < \frac{1}{\left [ \mathrm{max} \left ( \left | \mathbf{y}(0) \right | \right ) \right ]^2}$, $\gamma = 1+\sum_{j \in N(i)}w_{ij} > 0$ because for any $j \in N(i)$, $w_{ij} = \beta y_i(t)y_j(t) + 1 > 0$.

For any $t>1$, $y_i(t+1) $ is updated as in Equation~(\ref{al:updatePosDemon}), however, the sums have different values: $\sum_{j \in N(i)^s}w_{ij}\frac{y_j(t)}{y_i(t)}$, $\sum_{j \in N(i)^s}w_{ij}$, $\sum_{k \in N(i)^d}w_{ik}>0$, and $\sum_{k \in N(i)^d}w_{ik}\frac{y_k(t)}{y_i(t)}<0$.

Then we focus on the most opinionated node, which means the node has the largest absolution value of its opinion $\mathrm{max} \left | \mathbf{y}(t) \right |$, starting from time $0$. Knowing $D>0$, 
\begin{itemize}
\item when $C>0$, $C-D$ is shown in Equation~(\ref{al:diffCD}). With $i$ being the most opinionated node, $\frac{y_j(t)}{y_i(t)} \leq 1$ for all $j \in N(i)^s$; $\frac{y_k(t)}{y_i(t)}<0$ for all $k \in N(i)^d$. Therefore, $C<D$, $0<\frac{C}{D}<1$ and $\left | y_i(t+1) \right | < \left | y_i(t)\right |$.
\item when $C = 0$, $y_i(t+1) = 0$.
\item when $C<0$, $-C-D$ is shown in Equation~(\ref{al:diffNCD}). As $-1 \leq \frac{y_k(t)}{y_i(t)} \leq 0$ for $k \in N(i)^d$, $-C - D< 0$, $0< \left | \frac{C}{D}\right |<1$, thus $\left | y_i(t+1) \right | < \left | y_i(t)\right |$.
\end{itemize}
\begin{align}
-2 - \sum_{j \in N(i)^s}w_{ij}(\frac{y_j(t)}{y_i(t)}+1) - \sum_{k \in N(i)^d}w_{ik}(\frac{y_k(t)}{y_i(t)}+1) \label{al:diffNCD}
\end{align}

At every time step, the most opinionated node get moderated until they reach consensus - there is no such node and the updating process stops because consensus is reached.

\end{proof}

\begin{lemma}
For node $i \in V_1$, $y_i(0) = y_0$, where $0<y_0<1$; $\forall i \in V_2$, $y_i(0) = -y_0$. If $\beta = \frac{1}{y_0^2}$, $y_i(t)  = y_i(0)$ for all $t \geq 0$.
\end{lemma}
\begin{proof}
When $\beta = \frac{1}{y_0^2}$, $w_{ij} = \frac{1}{y_0^2} y_i(t) y_j (t)$. At time $1$,
\begin{align*}
y_i(1) = \frac{y_i(0) + 2n^s_i(0)y_i(0)}{  1 + 2n^s_i(0) }  = y_i(0)
\end{align*}

For any $t \geq 1$, 
$$y_i(t+1) = \frac{y_i(t) + 2n^s_i(t)y_i(t)}{  1 + 2n^s_i(t) }  = y_i(t) = y_i(0)$$.
\end{proof}

\section{Datasets and experimental results} \label{furtherExpRes}
\subsection{Real-world datasets}
\begin{table}[h!]
	\centering
	\caption{Real-world dataset summary}
	\label{tab:realDataset}
	\begin{tabular}{|c|c|c|c|}
		\hline
		Network    & $\left | V\right |$ & $\left | E\right |$ & Event                         \\ \hline
		Karate     & 34    & 78    & Friendship                    \\ \hline
		Tw:Club   & 703   & 3322  & Barcelona in La-liga 2016     \\ \hline
		Tw:Sport & 703   & 3322  & Juventus vs Real Madrid 2015 \\ \hline
		Tw:US     & 533   & 13564 & US Presidential Election 2016 \\ \hline
		Tw:UK     & 231   & 905   & British Election 2015         \\ \hline
		Tw:Delhi  & 548   & 3638  & Delhi Assembly Election 2013  \\ \hline
		Tw:GoT    & 947   & 7922  & GoT promotion 2015            \\ \hline
	\end{tabular}
\end{table}

\subsection{Influence of the opinion vector $\mathbf{y}(0)$ and network topology $G$}
\begin{figure}[h!]
\centering
\includegraphics[width = 1.0\hsize]{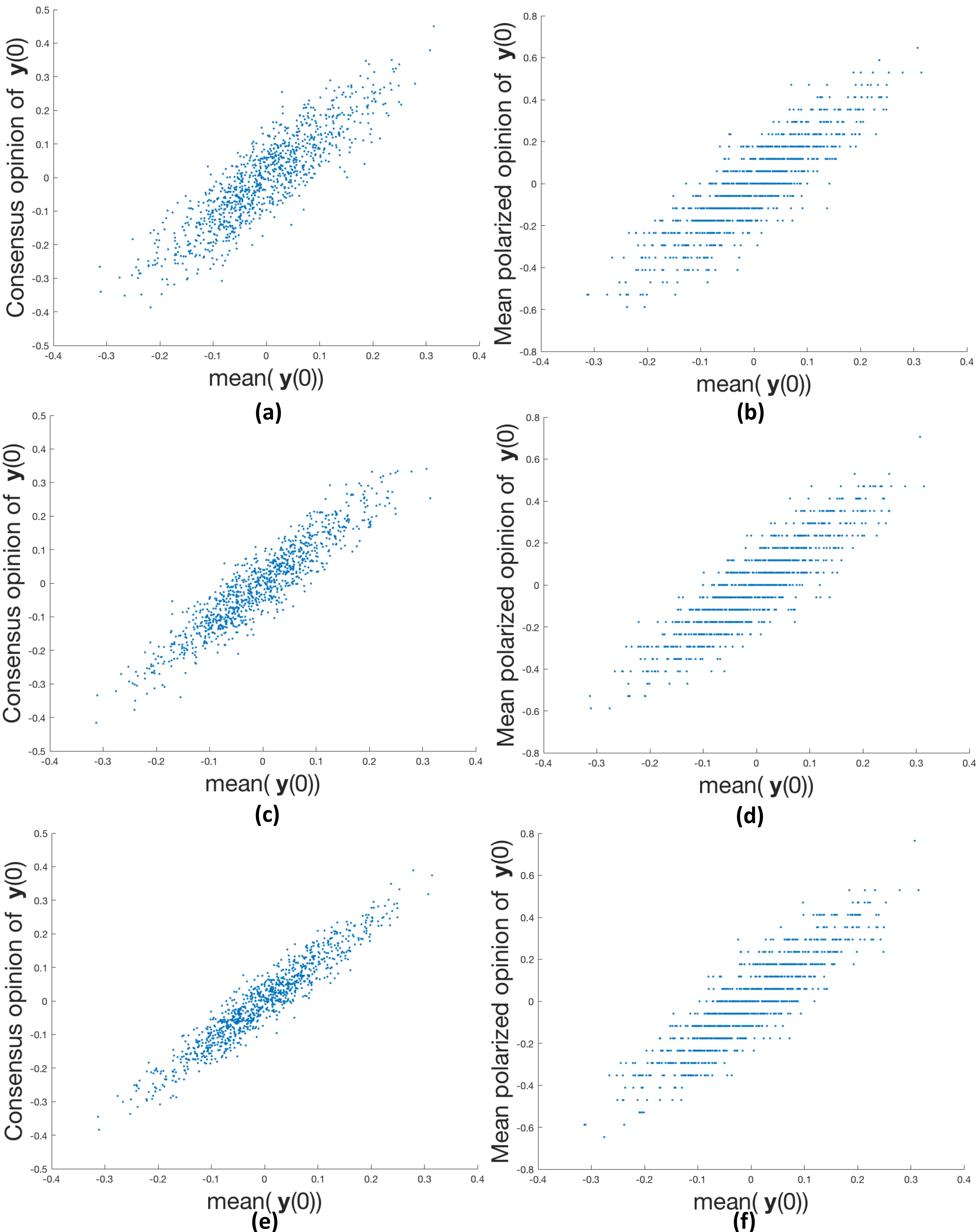}
\caption{For $1000$ random $\mathbf{y}(0)$. (a) and (b) on a BA model ($n = 34, M_0 = 3, M = 2$); (c) and (d) on an ER model ($n = 34, \rho = 0.139$); (e) and (f) on a WS model ($n = 34, K = 2$). The left column of (a), (c), (e) - the consensus opinion when $\beta = 1$; the right column of (b), (d), (f) - the mean polarized opinion when $\beta = 10$.}
\label{fig:E2Rest}
\end{figure}

This figure corresponds to Figure~\ref{fig:E2Result}, and is used to investigate both the effects of the opinion vector and the network topology. Horizontal subfigures show the different consensus and polarization converging states for different $\mathbf{y}(0)$s, while the vertical subfigures show the differences between the three types of random networks of similar sizes. The finding of this experiment is consistent with that of Figure~\ref{fig:E31Result}(b). 

\subsection{Influence of model parameters}
\begin{figure}[H]
\centering
\includegraphics[width = 1.0\hsize]{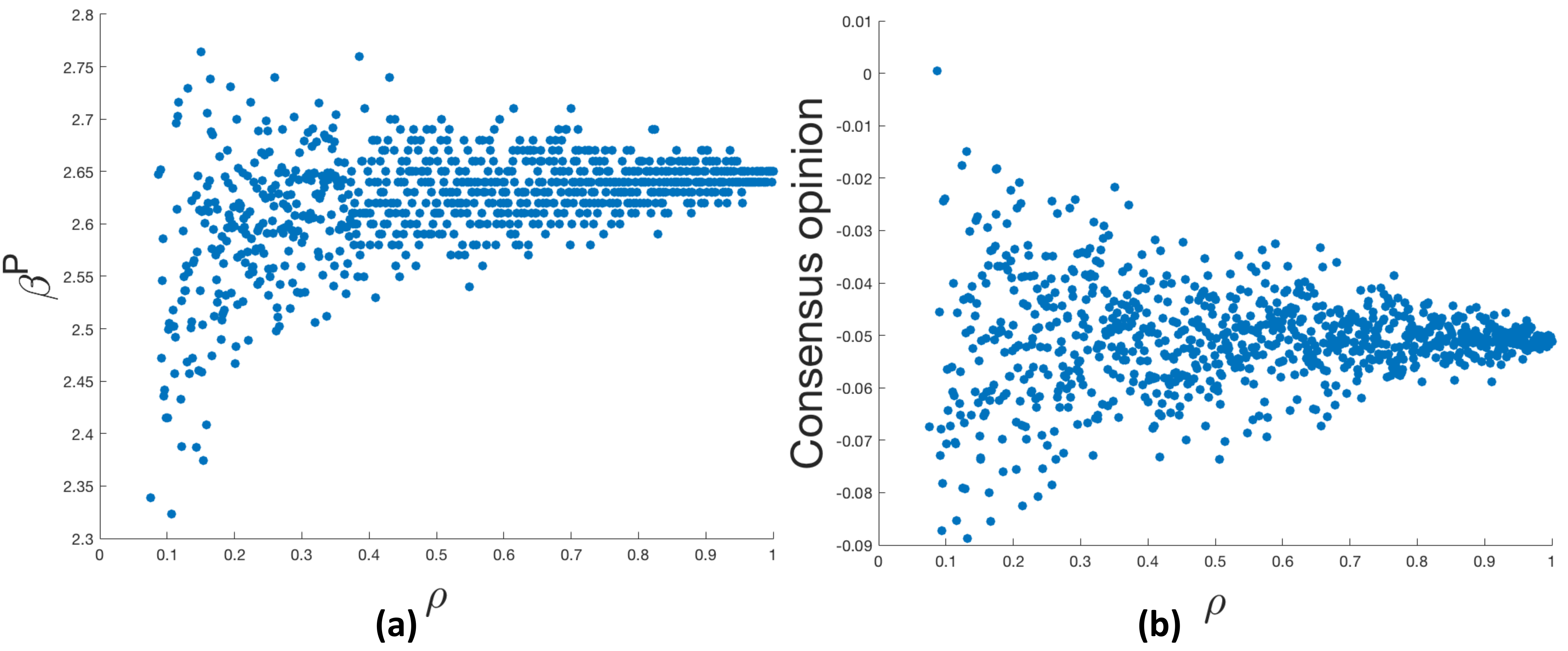}
\caption{For an random opinion vector $\mathbf{y}(0)$, on ER models with $n = 100$ and $\rho \in (0,1]$. (a) the value of $\beta^P$ for the $\mathbf{y}(0)$; (b) the consensus opinion reach by $\mathbf{y}(0)$ when $\beta = 1$.}
\label{fig:E32Result}
\end{figure}

This experiment takes the ER model as an example and investigates the influence of the parameter $\rho$ on the network topology, thus resulting in the influence on opinion dynamics. 

\subsection{Influence of edge placements}
\begin{figure}[H]
\centering
\includegraphics[width = 1.0\hsize]{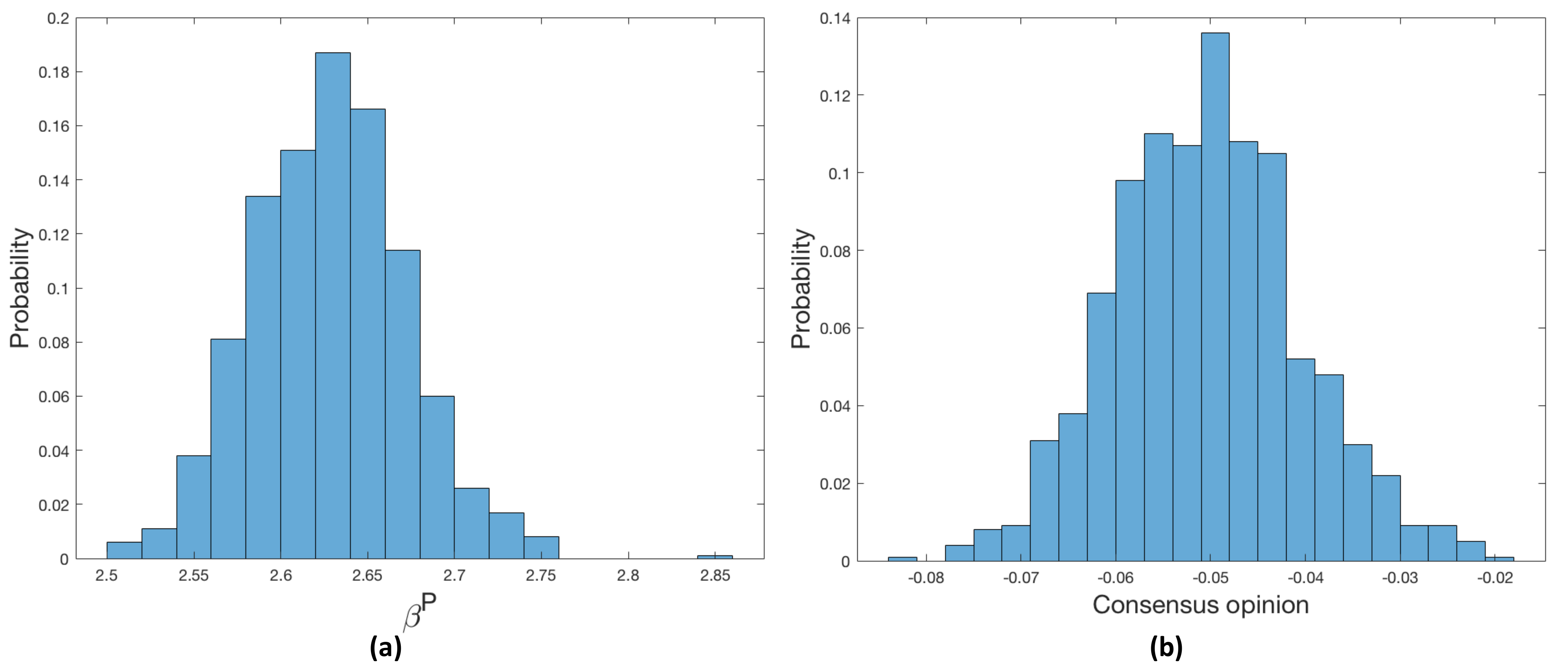}
\caption{For an random opinion vector $\mathbf{y}(0)$ with mean $-0.0395$, on $1000$ ER models with $n = 100$ and $\rho = 0.4 $. (a) the value of $\beta^P$ for the $\mathbf{y}(0)$; (b) the consensus opinion reach by $\mathbf{y}(0)$ when $\beta = 1$.}
\label{fig:E33Result}
\end{figure}

ER model is taken again as the example here for investigating the influence of the network edge placements on opinion dynamics. It shows that the network topology does have significant influence on the value of $\beta^P$ and the consensus opinion value.

\subsection{Influence of the edge addition/deletion in the network}
We can also investigate the question: If someone wants to maximally increase/decrease the value of consensus opinion or the average polarized opinion, which edge should be removed/added?

\ptitle{Add One Edge - Consensus}
\begin{figure}[H]
\centering
\includegraphics[width = 1.0\hsize]{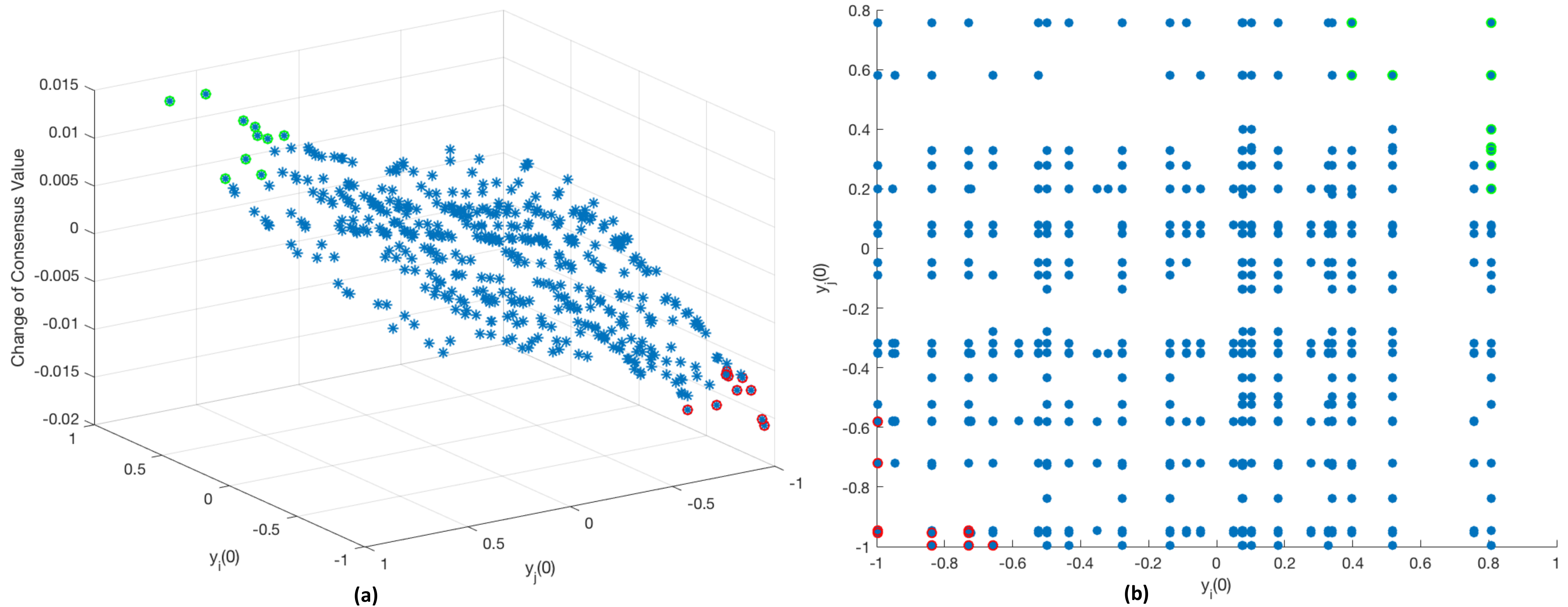}
\caption{Add one edge on Karate network to change the consensus opinion - $\beta = 1$. Top 10 best choices are highlighted: green for increase and red for decrease.}
\label{fig:E3KarateAddC}
\end{figure}

\ptitle{Delete One Edge - Consensus}
\begin{figure}[H]
\centering
\includegraphics[width = 1.0\hsize]{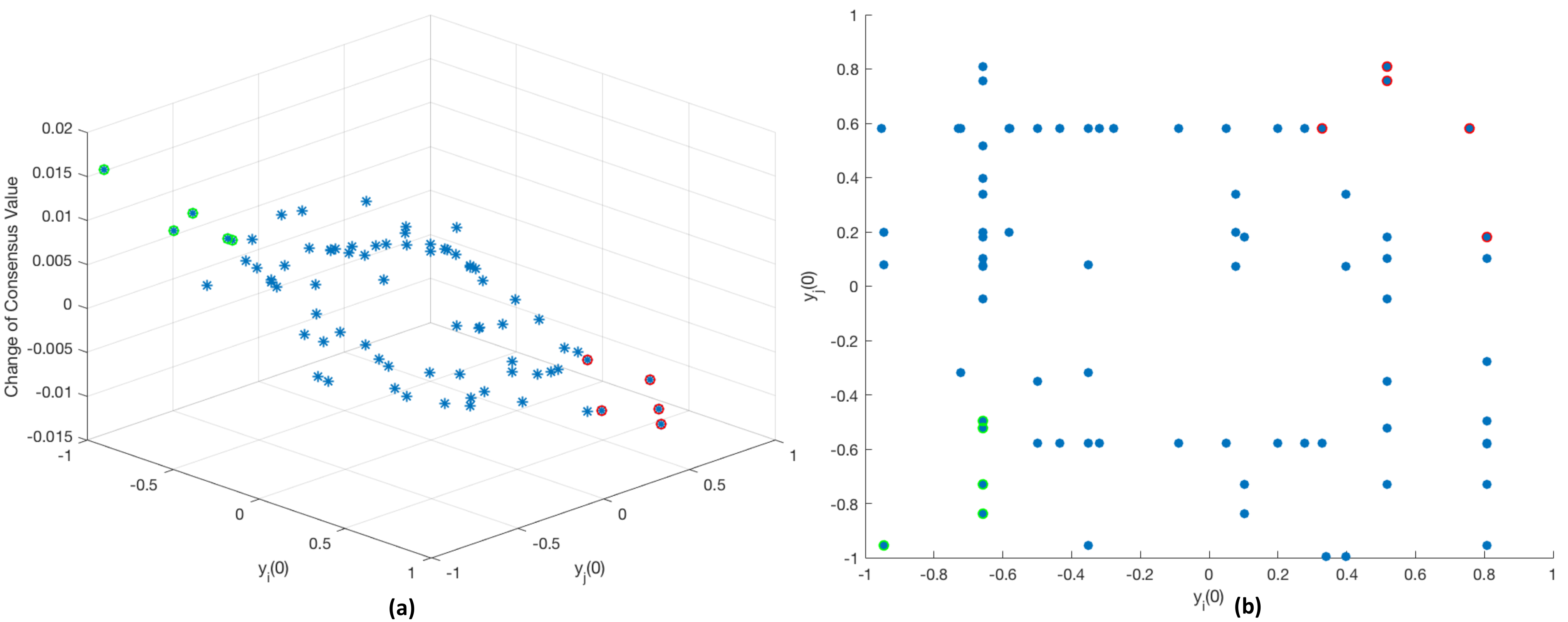}
\caption{Delete one edge on Karate network to change the consensus opinion - $\beta = 1$. Top 5 best choices are highlighted: green for increase and red for decrease.}
\label{fig:E3KarateDelC}
\end{figure}

It shows in Figure~\ref{fig:E3KarateAddC} and~\ref{fig:E3KarateDelC} that in order to maximally increase the consensus value by editing one edge, adding the edge between the most opinionated disconnected negative nodes is the best choice when allowing only addition; while deleting the edge between the most opinionated connected positive nodes is the most effective way if allowing only deletion. A smaller consensus opinion value can be achieved by adding the edge between the most positive opinionated nodes or deleting the one between the most negatively opinionated nodes.

\ptitle{Edge edition that has almost no influence on consensus}
\begin{figure}[H]
\centering
\includegraphics[width = 0.8\hsize]{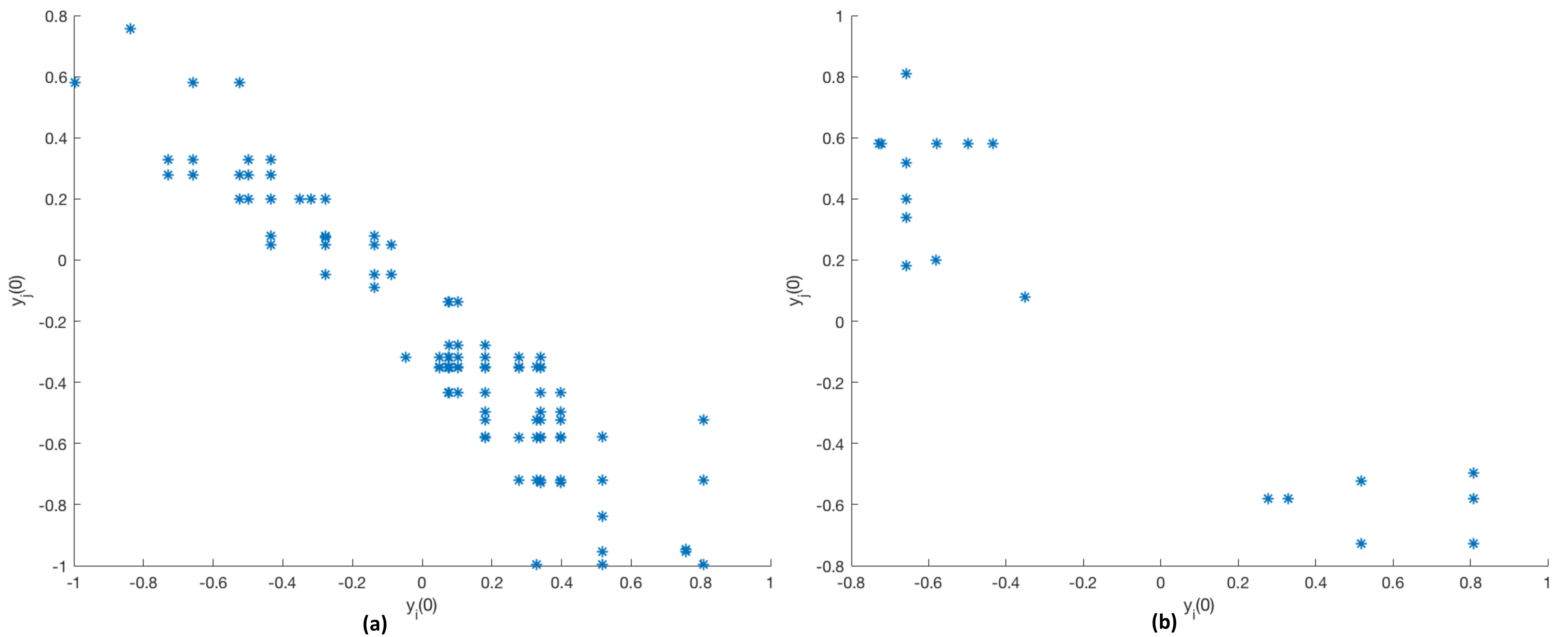}
\caption{Additions - (a) and Deletions - (b) that cause minor change in consensus values on Karate network. ($\left | \mathrm{change} \right | < 10^{-3}$)}
\label{fig:E3KarateCchange0}
\end{figure}

However, the connections between nodes with equivalently (i.e., in terms of absolute opinion value) opposing opinions have almost no influence on the consensus value, as shown in Figure~\ref{fig:E3KarateCchange0}. In contrast, when the network gets polarized, the neighbors of the neutral nodes have more significant influence on the mean polarized opinions.

\end{document}